\documentclass[sigconf]{aamas}  

\usepackage{booktabs}
\usepackage{tikz}
\usetikzlibrary{shapes.geometric}
\usepackage{algorithm}
\usepackage{algpseudocode}
\usepackage{xspace}
\usepackage[labelformat=parens,subrefformat=parens]{subcaption}
\def\trimafterfigure{\vspace{-0mm}}
\setcopyright{ifaamas}  
\acmDOI{doi}  
\acmISBN{}  
\acmConference[]{}{}{}  
\acmYear{}  
\copyrightyear{}  
\acmPrice{}  

\theoremstyle{plain}
\newtheorem{observation}{Observation}

\begin{document}

\title{The Complexity of Max-Min $k$-Partitioning}

\author{Anisse Ismaili\vspace*{5mm}}

\begin{abstract}  
In this paper we study a max-min $k$-partition problem on a weighted graph,
that could model a robust $k$-coalition formation.
We settle the computational complexity of this problem as complete for class $\Sigma_2^P$. This hardness holds even for $k=2$ and arbitrary weights, or $k=3$ and non-negative weights, which matches what was known on \textsc{MaxCut} and \textsc{Min-3-Cut} one level higher in the polynomial hierarchy. 
\end{abstract}


\maketitle

\vspace*{5mm}

\section{Preliminaries}
A max-min $k$-partition instance is defined by $\langle N, L, w, k, m, \theta \rangle$.
\begin{itemize}
      \item $(N, L, w)$ is a weighted undirected graph.
  $N=[n]$, where $n \in \mathbb{N}$ 
  is a set of nodes.\footnote{Given $n \in \mathbb{N}$, $[n]$ is shorthand of $\{1, \ldots, n\}$.}
    The set of links $L \subseteq {N \choose 2}$ consists of unordered node pairs.
    Link $\ell=\{i,j\}$ maps to weight $w_{ij} \in \mathbb{Z}$.
    {Equivalently, $w:N^2\rightarrow\mathbb{Z}$ satisfies for any $(i,j)\in N^2$ that $w(i,i)=0$, $w(i,j)=w(j,i)$ and $w(i,j)\neq 0\Rightarrow\{i,j\}\in L$.}
   \item $k$ is the size of a partition, $2 \le k < n$.
   \item $m \in \mathbb{N}$ is the number of nodes that could be removed.
   \item $\theta \in \mathbb{Z}$ is a threshold value. 
\end{itemize}
Let $\pi$ denote a $k$-partition of $N$, which is a collection of node-subsets
$\{S_1, \ldots, S_k\}$,
such that for each $i \in [k]$, $S_i \subseteq N$, and 
$\forall S_i, S_j\in\pi$, where $i\neq j$, $S_i\cap S_j=\emptyset$ holds. 
We say that a $k$-partition $\pi$ is complete when
$\bigcup_{i \in [k]}  S_i = N$ holds (otherwise, it is incomplete). 
For a complete partition $\pi$ and an incomplete partition $\pi'$,
we say that $\pi$ subsumes $\pi'$ when
$S_i \supseteq S'_i$ holds for all $i \in [k]$. 
For node $i\in N$, $\pi(i)$ is the node-subset to which it belongs.
For any $S \subseteq N$,
define  
$$W(S)=\sum\nolimits_{\{i,j\}\subseteq S}w(i,j).$$
Then, let $W(\pi)$ denote $\sum_{S \in \pi} W(S)$. We require that no
node-subset be empty; hence, if some node-subset is empty, we set $W(\pi)=-\infty$.

Given a $k$-partition $\pi=\{S_1, \ldots, S_k\}$ and a set $M \subseteq N$, 
the remaining incomplete partition $\pi_{-M}$ after removing $M$ is 
defined as $\{S'_1, \ldots, S'_k\}$, where $S'_i = S_i \setminus M$.
Let $W_{-m}(\pi)$ denote the minimum value after removing at most $m$ nodes,
i.e., it is defined as: 
$$W_{-m}(\pi)=\min\limits_{M\subseteq N,|M|\leq m}\{W(\pi_{-M})\}.$$
To obtain $W_{-m}(\pi)\neq-\infty$, 
every $S\in\pi$ needs to contain at least $m+1$ nodes,
so that no node-subset of $\pi_{-M}$ is emptied.
For partition $\pi=\{S_1, \ldots, S_k\}$, we define its deficit count
$\text{df}(\pi)$ as $\sum_{i\in[k]} \max(0, m+1 - |S_i|)$. 
Thus, $\text{df}(\pi)=0$ must hold in order to obtain
$W_{-m}(\pi)\neq-\infty$. 


\begin{definition}
The decision version (1) of our main problem is defined below. It may also be referred to as the defender's problem.
\begin{enumerate}
\item \textsc{Max-Min-$k$-Partition}:
Given a max-min $k$-partition instance, 
is there any $k$-partition $\pi$
satisfying $W_{-m}(\pi)\geq \theta$?
\item \textsc{Max-Min-$k$-Partition/Verif}:      
Given an instance of a max-min $k$-partition and
a partition $\pi$,
does $W_{-m}(\pi)\geq \theta$ hold?
\end{enumerate}
A key step is to study the natural verification problem (2), to which complement we refer as the attacker's problem. (Does an attack $M\subseteq N,|M|\leq m$ on $\pi$ exist such that $W(\pi_{-M})\leq\theta-1$?)
\end{definition}


\vspace*{5mm}

\section{Complexity of \textsc{Max-Min-$k$-Partition}}
\label{sec:complexity}

\vspace*{2mm}

In this section, we address the computational complexity of the defender's problem.
The verification (resp. attacker's) problem itself turns out to be coNP-complete (resp. NP-complete), 
which intricates one more level in the polynomial hierarchy (PH).  
We show that \textsc{Max-Min-$k$-Partition} is complete for class $\Sigma_2^P$, even in two cases:
\begin{enumerate}
\item[(a)] when $k=2$ for arbitrary link weights $w\lessgtr 0$, or
\item[(b)] when $k=3$ for non-negative link weights $w\geq 0$.
\end{enumerate}
These results seem to match what was known on 
\textsc{MaxCut} \cite{karp1972reducibility} (contained in \textsc{Min-2-Cut} when $w\lessgtr 0$ and NP-complete) and 
\textsc{Min-3-Cut} \cite{Dahlhaus:1992:CMC:129712.129736} (NP-complete for $w\geq 0$ when one node is fixed in each node-subset), 
but one level higher in PH.

\begin{observation}\label{rk:1}
\textsc{Max-Min-$k$-Partition/Verif} is coNP-complete.
It holds even for $k=1$, weights $w$ in $\{0,1\}$ and threshold $\theta=1$.
\end{observation}

\begin{proof}
Decision problem \textsc{Max-Min-$k$-Partition/Verif} is in class coNP, since for any no-instance, a failing set $M$ such that $W(\pi_{-M})\leq\theta-1$ is a no-certificate verifiable in polynomial-time.

We show coNP-hardness by reduction from \textsc{MinVertexCover} to the (complement) attacker's problem.
Let graph $G=(V,E)$ and vertex number $m\in\mathbb{N}$ be any instance of \textsc{MinVertexCover}.
\textsc{MinVertexCover} asks whether there exists a vertex-subset $U\subseteq V,|U|\leq m$ such that $\forall \{i,j\}\in E, i\in U\mbox{ or } j\in U$, i.e. every edge is covered by a vertex in $U$. We reduce it to an attacker's instance with nodes $N\equiv V$, weights $w(i,j)\in\{0,1\}$ equal to one if and only if $\{i,j\}\in E$ and threshold $\theta=1$. The verified partition is simply $\pi=\{N\}$. The idea is that constraint $W(\pi_{-M})\leq 0$ is equivalent to damaging every link, hence to finding a vertex-cover $U\equiv M$ with $|M|\leq m$.
%
%
%
\end{proof}

We now proceed with the computational complexity of the main defender's problem under $w\lessgtr 0$ and $w\geq 0$.
We show $\Pi_2^P$-hardness of the $\forall\exists$ complement by reduction from \textsc{MaxMinVertexCover} or \textsc{$\forall\exists$3SAT}.
The idea is to (1) enforce that only some \emph{proper} partitions are meaningful. One possible proper partition corresponds to one choice on $\forall$ in the original problem. Then, (2) within one particular node-subset of a proper partition, we represent the subproblem (e.g. \textsc{VertexCover} or $\textsc{3-SAT}\leq \textsc{IndependentSet}=\textsc{VertexCover}$).

\begin{theorem}\label{th:k2}
Problem \textsc{Max-Min-$k$-Partition} is $\Sigma_2^P$-complete, 
even for $k=2$ node-subsets and $w\in\{-n^2,1,2\}$.
\end{theorem}
\begin{figure}[b]
\centering
\begin{tikzpicture}[scale=0.8]
\node[thick, rectangle,draw=black!50!green, inner sep=0.5em] (n10) at (0,0) {$N_{1,0}$};
\node[thick, rectangle,draw=black!50!green, inner sep=0.5em] (n20) at (2.4,0) {$N_{2,0}$};
\node[thick, rectangle,draw=black!50!blue, inner sep=0.5em] (n30) at (4.8,0) {$N_{3,0}$};
\node[thick, rectangle,draw=black!50!green, inner sep=0.5em] (n40) at (7.2,0) {$N_{|I|,0}$};
\node[thick, rectangle,draw=black!50!blue, inner sep=0.5em] (n11) at (0,-1.6) {$N_{1,1}$};
\node[thick, rectangle,draw=black!50!blue, inner sep=0.5em] (n21) at (2.4,-1.6) {$N_{2,1}$};
\node[thick, rectangle,draw=black!50!green, inner sep=0.5em] (n31) at (4.8,-1.6) {$N_{3,1}$};
\node[thick, rectangle,draw=black!50!blue, inner sep=0.5em] (n41) at (7.2,-1.6) {$N_{|I|,1}$};
\draw[]
	(n10) edge[dashed] node[left] {$-\Lambda$} (n11)
	(n20) edge[dashed] node[left] {$-\Lambda$} (n21)
	(n30) edge[dashed] node[right] {$-\Lambda$} (n31)
	(n40) edge[dashed] node[right=2mm] {$-\Lambda$} (n41);
\draw[]
	(n10) edge[black!30,thin, out=310, in=170]  (n21)
	(n10) edge[black!30,thin, out=20, in=105]  (n31)
	(n10) edge[black!30,thin] (n20)
	(n11) edge[black!30,thin] (n21)
	(n20) edge[black!30,thin] (n30)
	(n20) edge[black!30,thin] (n31)
	(n20) edge[black!30,thin, out=15, in=165] (n40)
	(n21) edge[black!30,thin] (n31)
	(n21) edge[black!30,thin, out=340, in=255] (n40)
	(n11) edge[black!30,thin, out=340, in=200] (n31)
	(n30) edge[black!30,thin] (n40)
	(n31) edge[black!30,thin] (n40)
	(n31) edge[black!30,thin] (n41);
\draw [thick, black!35!green] plot [smooth, tension=0.5] coordinates 
	{ (-0.7,0) (0,-0.5) (2.4,-0.5) (4.8,-2.2) (7.9,-0.3) (7.1,0.6) (4.8,-1.2) (2.6,0.5) (0,0.5) (-0.7,0)};
\draw [thick, blue] plot [smooth, tension=0.5] coordinates 
	{ (-0.7,-1.6) (0,-1.1) (2.4,-1.1) (4.8,0.6) (7.9,-1.3) (7.1,-2.2) (4.8,-0.4) (2.6,-2.1) (0,-2.1) (-0.7,-1.6)};
\end{tikzpicture}
\vspace*{-2mm}
\caption{Reduction from \textsc{MaxMinVertexCover} to \textsc{co-Max-Min-$k$-Partition}: $w_{ij}=2$ if and only if $\{ij\}$ is an edge. A proper 2-partition $\pi=\{S_1,S_2\}$ is in green ($S_1$) and blue ($S_2$).}\label{fig:th:2}
\end{figure}
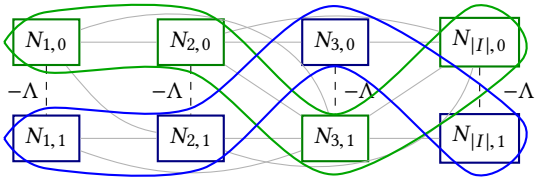
\begin{proof}
Decision problem \textsc{Max-Min-$k$-Partition} asks whether 
$\exists~k\mbox{-partition } \pi,
\forall~M\subseteq N, |M|\leq m,
W(\pi_{-M})\geq\theta$.
Therefore, it lies in class $\Sigma_2^P$, since, for yes-instances, such a $k$-partition $\pi$ is a certificate that can be verified by an NP-oracle on the remaining coNP problem \textsc{Max-Min-$k$-Partition/Verif}. We show $\Sigma_2^P$-hardness by a (complementary) reduction from $\Pi_2^P$-complete problem \textsc{MaxMinVertexCover}, defined as follows. Given graph $G=(V,E)$ whose vertices are partitioned by  index set $I$ into $V=\bigcup_{i\in I}(V_{i,0}\cup V_{i,1})$, for a function $p:I\rightarrow\{0,1\}$, we define $V^{(p)}=\bigcup_{i\in I}V_{i,p(i)}$ and induced subgraph $G^{(p)}=(V^{(p)},E^{(p)})$. Given $m\in\mathbb{N}$, it asks whether:
$$
\forall p\!:\!I\!\rightarrow\!\{0,1\},\quad
\exists U\subseteq V^{(p)}\!\!,|U|\!\leq\! m,\quad
U\mbox{ is a vertex cover of } G^{(p)}.
$$
where ``$U$ is a vertex cover of $G^{(p)}$'' means $\forall \{u,v\}\in E[V^{(p)}]$, $u\in U$ or $v\in U$.
Since edges between $V_{i,0}$ and $V_{i,1}$ are never relevant, we can remove them. By \cite[Th. 10, proof]{Ko1995}, all $V_{i,j}$ sets have the same size, hence set $V^{(p)}$ has a constant size $n$ for any $p$.

The reduction is described in Figure \ref{fig:th:2}.
We reduce any instance of \textsc{MaxMinVertexCover} (as described above) to the following \emph{complementary} instance of \textsc{Max-Min-$k$-Partition}.
Nodes $N\equiv V$ are identified with vertices, hence can also be partitioned by $I\times\{0,1\}$ into $N=\bigcup_{i\in I}(N_{i,0}\cup N_{i,1})$ with $N_{i,j}\equiv V_{i,j}$ . We ask for $k=2$ node-subsets and choose a large number $\Lambda$, e.g. $\Lambda=n^2$. For every link $\{i,j\}\in{N \choose 2}$, if $\{i,j\}\in E$, we define synergy $w(i,j)=2$; otherwise if $\{i,j\}\notin E$, we define $w(i,j)=1$. However, for every $\ell\in I$ and every $(i,j)\in N_{\ell,0}\times N_{\ell,1}$, we define negative weight $w(i,j)=-\Lambda$. Here, up to $2m$ nodes might fail, and threshold $\theta=f_{n,m}(m)+1$ is defined in the proof.
Since we are working on a complementary instance, the question is whether
$$
\forall 2\mbox{-partition }\pi,\quad
\exists M\subseteq N,|M|\leq 2m,\quad
W(\pi_{-M})\leq f_{n,m}(m),
$$
where $f_{n,m}:[0,2m]\rightarrow[0,n^2]$ is defined later.

This condition is trivially satisfied on 2-partitions $\pi$ where for some $\ell\in[I]$, two nodes $(i,j)\in N_{\ell,0}\times N_{\ell,1}$ are in the same node-subset. Indeed, even with an empty attack $M=\emptyset$, weight $W(\pi_{-\emptyset})$ incurs synergy $w(i,j)=-\Lambda$ and $W(\pi_{-\emptyset})<0\leq f_{n,m}(m)$.
Therefore, the interesting part of this condition is on the other 2-partitions: the \emph{proper} 2-partitions $\pi=\{S_1,S_2\}$, which satisfy $\forall\ell\in[I],\forall (i,j)\in N_{\ell,0}\times N_{\ell,1}, \pi(i)\neq\pi(j)$.
It's easy to see that $\pi$ can be characterized by a function $p\!:\!I\!\rightarrow\!\{0,1\}$ such that  $S_1=\bigcup_{i\in I}N_{i,p(i)}$ and $S_2=\bigcup_{i\in I}N_{i,1-p(i)}$, and $|S_1|=|S_2|=n$.
Since the remaining weights inside $S_1$ and $S_2$ are positive, the largest failures are the most damaging, $|M|=2m$ holds.

We now define function $f_{n,m}$. It maps $x\in[0,2m]$ to the number of in-subset pairs in a proper 2-partition $\pi=\{S_1,S_2\}$ ($|S_1|\!=\!|S_2|\!=\!n$) after $x$ nodes fail in $S_1$ and $2m-x$ in $S_2$ (total $2m$ failures). One has:
$$
f_{n,m}(x)
\enskip=\enskip 2{n\choose 2}-\sum\limits_{i=1}^{x}(n\!-\!i)-\sum\limits_{j=1}^{2m-x}(n\!-\!j)
\enskip=\enskip g_{n,m}+ x(x-2m),
$$
where $g_{n,m}$ is constant w.r.t. $x$.
Since $f'_{n,m}(x)=2(x-m)$ and $f''_{n,m}(x)=2$,
it is a strictly convex function with minimum point at $x=m$.
Therefore, for integers $x\in[2m]$, if $x\neq m$, the inequality $f_{n,m}(x) > f_{n,m}(m)$ holds.
By definition, $f_{n,m}(x)$ is a lower bound on $W(\pi_{-M})$ (by assuming that all remaining weights in $\pi_{-M}$ have a value of $1$, instead of $1$ or $2$). Therefore, the main condition can only be satisfied by \emph{balanced} failures $M=M_1\cup M_2$ such that $M_1\subseteq S_1$, $M_2\subseteq S_2$ and crucially: $|M_1|=|M_2|=m$. 

(yes$\Rightarrow$yes) Any subgraph $G^{(p)}$ admits a vertex cover $U\subseteq V^{(p)}$ with size $|U|\leq m$. Let us show that any proper 2-partition $\pi=\{S_1,S_2\}$ (characterized by a function $p:I\rightarrow\{0,1\}$) can be failed down to $f_{n,m}(m)$. Let $M_1\subseteq S_1$ correspond to the vertex cover of subgraph $G^{(p)}$ and $M_2\subseteq S_2$ to the vertex cover of subgraph $G^{(1-p)}$. Then, the failing set $M=M_1\cup M_2$ has a size of $|M|\leq 2m$, is balanced, and any node pair $\{i,j\}$ of weight two in $\pi$ (edge in $E$) has $i$ or $j$ in $M$, by the vertex covers. All in all, $W(\pi_{-M})=f_{n,m}(m)$.

(yes$\Leftarrow$yes) Any proper 2-partition $\pi=\{S_1,S_2\}$ (characterized by function $p:I\rightarrow\{0,1\}$) admits a well balanced failing set $M=M_1\cup M_2$ such that $W(\pi_{-M})\!\leq\!f_{n,m}(m)$. Then it must be the case that $M_1$ (and $M_2$) covers all the node pairs of synergy two in $S_1$ (resp. $S_2$) that correspond to the edges of $G^{(p)}$ (resp. $G^{(1-p)}$). 
Then, for any subgraph $G^{(p)}$, attack $U\equiv M_1$ is a vertex cover.
\end{proof}

Adding a constant to all weights  does not preserve optimal solutions.
Thus, we cannot modify a problem with negative weights to an
equivalent non-negative weight problem.
Still, a hardness result for $k=3$ can also be obtained 
from \textsc{$\forall\exists$3SAT}. 


\begin{theorem}\label{th:k3}
\textsc{Max-Min-$k$-Partition} is $\Sigma_2^P$-complete, 
even for $k=3$ node-subsets and weights $w\in\{0,\Lambda,\Lambda+1\}$, where $\Lambda\geq n^2$.
\end{theorem}

\begin{proof}
Let us first recall a classical reduction from  \textsc{3SAT} to \textsc{IndependentSet}, and how the later relates to \textsc{VertexCover}.
Let any 3SAT instance be defined by formula $F=C_1\wedge\ldots\wedge C_{\alpha}$, where $C_i$ is a 3-clause on variables $X$. Every clause $C_i=\ell_{i,1}\vee\ell_{i,2}\vee\ell_{i,3}$ is reduced to triangle of vertices $V_i=\{v_{i,1},v_{i,2},v_{i,3}\}$ representing the literals of the clause. The set of $3\alpha$ vertices is then $V=\cup_{i=1}^{\alpha} V_i$. Between any two subsets $V_i,V_j$, edges exist between two vertices if and only if the corresponding literals are on the same variable and are complementary (hence incompatible). It is easy to see that an independent-set $U\subseteq V$ of size $\alpha$ must have exactly one vertex per triangle $V_i$, and will exist (no edges within) if and only if there exists an instantiation of $X$ that makes at least one literal per clause $C_i$ true.
Given a graph $G=(V,E)$, if $U\subseteq V$ is an independent-set, 
it means that 
$i\in U\wedge j\in U\Rightarrow \{i,j\}\notin E$.
Hence, contraposition 
$\{i,j\}\in E\Rightarrow (i\!\in\!V\!\setminus\!U)\vee (j\!\in\!V\!\setminus\!U)$
means that  $V\!\setminus\!U$ is a vertex cover.
For instance, in the reduction from \textsc{3SAT}, one can equivalently ask for a vertex cover $V\setminus U$ with size $2\alpha$; that is, two vertices per triangle $V_i$:
Set $V$ of third vertices shall have no edge left to cover.

Let any instance of \textsc{$\forall\exists$3SAT} be defined by 3CNF formula $F(X,Y)=\bigwedge_{i=1}^{\alpha}C_i$ on variables $X=\{x_1,\ldots,x_{|X|}\}$ and $Y=\{y_1,\ldots,y_{|Y|}\}$. This problem asks whether:
$$
\forall \tau_x:X\rightarrow\{0,1\},\quad
\exists \tau_y:Y\rightarrow\{0,1\},\quad
F(\tau_x,\tau_y)\text{ is true}.
$$
Without loss of generality, one can assume there is at most one $X$-literal per clause $C$.
Indeed, if there are three $X$-literals, some $\tau_x$ can make the clause false, and it is trivially a no-instance.
If there are two $X$-literals: $C=x\vee x'\vee y$, then by adding a fresh $Y$-variable $z$, one easily obtains $C=(x\vee z\vee y)\wedge(x'\vee \neg z\vee y)$.
For ease of presentation, we assume exactly one $X$-literal and two $Y$-literals. We extend this proof to including clauses with no $X$-literal, in its final remark. Let $X(C)$ be the $X$-\emph{literal} in clause $C$. 

We build a \textsc{Max-Min-3-Partition} instance on $n\!=\!10\alpha + 2$ nodes with  $m=2\alpha$ failures . We first describe the nodes.
To every clause $C_i=\ell^{x}_{i}\vee\ell^{y}_{i}\vee\ell^{y'}_{i}$, we associate two node tetrads $N_{i,0}=\{v^{x}_{i,0},v^{y}_{i,0},v^{y'}_{i,0},v^{z}_{i,0}\}$ and $N_{i,1}=\{v^{x}_{i,1},v^{y}_{i,1},v^{y'}_{i,1},v^{z}_{i,1}\}$ (both depicted in Figure \ref{fig:tetrads}) which represent the two scenarios on $X$-literal $\ell^{x}_{i}$: false or true. Hence, there are $2\alpha$ node tetrads and a total of $4m=8\alpha$ nodes in $T=\cup_{i=1}^{\alpha}\cup_{j\in\{0,1\}}N_{i,j}$. There is also a set $K$ of $m=2\alpha$ nodes, and two nodes $v^{1/2},v^{2/2}$. This construct is depicted in Figure \ref{fig:many-tetrads}.

\begin{figure}
\begin{tikzpicture}[square/.style={regular polygon,regular polygon sides=4}]
\node[circle, thin, draw=black, inner sep=0mm] (vzi0) at (0,1) {$v^z_{i,0}$};
\node[circle, thin, draw=black, inner sep=0mm] (vxi0) at (0,2) {$v^x_{i,0}$};
\node[circle, thin, draw=black, inner sep=0mm] (vy1i0) at (-1,0.5) {$v^{y}_{i,0}$};
\node[circle, thin, draw=black, inner sep=0mm] (vy2i0) at (1,0.5) {$v^{y'}_{i,0}$};
\node[circle, thin, draw=black, inner sep=0mm] (vzi1) at (5,1.5) {$v^z_{i,1}$};
\node[circle, thin, draw=black, inner sep=0mm] (vxi1) at (5,0.5) {$v^x_{i,1}$};
\node[circle, thin, draw=black, inner sep=0mm] (vy1i1) at (4,2) {$v^{y}_{i,1}$};
\node[circle, thin, draw=black, inner sep=-0.2mm] (vy2i1) at (6,2) {$v^{y'}_{i,1}$};
\draw[thick] (vxi0)--(vy1i0)--(vy2i0)--(vxi0);
\draw[ultra thick] (vzi0)--(vxi0);
\draw[thick] (vxi1)--(vy1i1)--(vy2i1)--(vxi1);
\node[text width=22mm] at (2.5,1.5) {No outgoing 1-link from $v^x_{i,j}$};
\node[square, draw=blue] at (5,1.5) {\hspace*{3mm}};
\node[square, draw=blue] at (5,0.5) {\hspace*{3mm}};
\node[square, draw=red] at (4,2) {\hspace*{3mm}};
\node[square, draw=red] at (6,2) {\hspace*{3mm}};
\node[square, draw=red] at (0,2) {\hspace*{3mm}};
\node[square, draw=blue] at (0,1) {\hspace*{3mm}};
\node[square, draw=red] at (-1,0.5) {\hspace*{3mm}};
\node[square, draw=blue,dashed] at (-1,0.5) {\hspace*{4.5mm}};
\node[square, draw=blue] at (1,0.5) {\hspace*{3mm}};
\node[square, draw=red,dashed] at (1,0.5) {\hspace*{4.5mm}};
\end{tikzpicture}
\vspace*{-3mm}
\caption{For clause $C_i=\ell^x_i\vee\ell^{y}_i\vee\ell^{y'}_i$, tetrads $N_{i,0}$ and $N_{i,1}$:\\Vertex-covers (red) and Independent-sets (blue) of size 2. Node $v^x_{i,0}$ (resp. $v^x_{i,1}$) is in no (resp. every) independent-set.}\label{fig:tetrads}
\trimafterfigure
\end{figure}
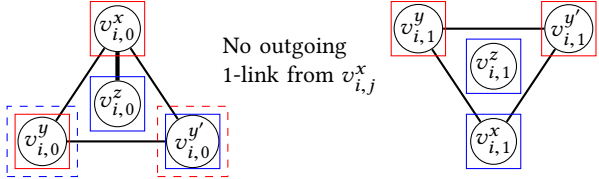

To describe the weights, we define a number $\Lambda\gg 1$, 
and only three different link weights $0,\Lambda,\Lambda+1$. 
We call $\Lambda$-link any link with weight $\Lambda$ or $\Lambda+1$.
We call 1-link any link with weight $\Lambda+1$.
Every pair of nodes in $\cup_{i=1}^{\alpha}\cup_{j\in\{0,1\}}N_{i,j}$ 
are linked by weight $\Lambda$ or $\Lambda+1$, 
except ($\star$) we set weights \emph{zero} (and no link) for every $i,i'\in[\alpha]$:
\begin{itemize}
\item[when] $X(C_i)=X(C_{i'})$ between $N_{i,j}$ and $N_{i',1-j}$ for $j\in\{0,1\}$, or 
\item[when] $X(C_i)= \neg X(C_{i'})$ between $N_{i,j}$ and $N_{i',j}$ for $j\in\{0,1\}$.
\end{itemize} 
The rationale is to forbid two inconsistent scenarios on a same $X$-variable to coexist in one node-subset.

Whether the $\Lambda$-link is also a 1-link is determined as follows.
Inside every node tetrad $N_{i,j}=\{v^{x}_{i,j},v^{y}_{i,j},v^{y'}_{i,j},v^{z}_{i,j}\}$,
there is a triangle of 1-links: 
$\{v^{x}_{i,j},v^{y}_{i,j}\}$,
$\{v^{y}_{i,j},v^{y'}_{i,j}\}$ and
$\{v^{y'}_{i,j},v^{x}_{i,j}\}$.
Only in negative tetrads $N_{i,0}$, there is a 1-link $\{v^{x}_{i,0},v^{z}_{i,0}\}$.
Given any tetrad $N_{i,j}$, node $v^{x}_{i,j}$ is not involved in any outgoing 1-link,
but only links with weight $\Lambda$.
Between any tetrads $N_{i,j}$ and $N_{i',j'}$ except ($\star$),
there is a 1-link between complementary nodes of $Y$-literals; that is, a 1-link exists when the later's literal is the negation of the former's.\footnote{It is the same idea as in the standard reduction from 3SAT to \textsc{IndependentSet}.}
Assuming w.l.o.g. that $\alpha$ is even,
let $\mu_1$ be the number of 1-links in $\bigcup_{i=1}^{i=\alpha/2} N_{i,0}$
and $\mu_2$ in $\bigcup_{i=\alpha/2+1}^{i=\alpha} N_{i,0}$.

Inside $K$, every pair of nodes is linked by weight $\Lambda$.
Also, every node in $K$ is linked to every node in tetrads $T$ by weight $\Lambda$.
Node $v^{1/2}$ is linked to every node in $\bigcup_{i=1}^{i={\alpha}/{2}}\bigcup_{j\in\{0,1\}}N_{i,j}$
by weight $\Lambda$, except for nodes $v^z_{i,1}$ by weight $\Lambda+1$; the same holds from node $v^{2/2}$
to every node in $\bigcup_{i={\alpha}/{2}+1}^{i=\alpha}\bigcup_{j\in\{0,1\}}N_{i,j}$.
All other weights are zeros. We achieve this construct by defining threshold $\theta$ as: 
$$
\theta-1 \quad=\quad {2m\choose 2}\Lambda + 2{m+1\choose 2}\Lambda + \mu_1 + \mu_2,
$$
and asking whether
$
\forall 3\text{-part }\pi,
\exists M\!\subseteq\!N,|M|\!\leq\!m,
W(\pi_{-M})\leq \theta\!-\!1.
$

A \emph{proper}-3-partition $\pi=\{S^{(p)},S^{1/2},S^{2/2}\}$ is characterized by an instantiation $p:X\rightarrow\{0,1\}$ of $X$ variables extended to literals by $p(\neg x)=1-p(x)$, and which defines:
$$
\begin{array}{cccclr}
S^{(p)} &=& K &\cup& \bigcup_{i=1}^{i=\alpha} N_{i,p(X(C_i))} & (3m\text{ nodes})\\
S^{1/2} &=& \{v^{1/2}\} &\cup& \bigcup_{i=1}^{i={\alpha}/{2}} N_{i,1-p(X(C_i))} & (m+1\text{ nodes})\\
S^{2/2} &=& \{v^{2/2}\} &\cup& \bigcup_{i={\alpha}/{2}+1}^{i=\alpha} N_{i,1-p(X(C_i))} & (m+1\text{ nodes})
\end{array}
$$
Note that in $S^{1/2}$ (resp. $S^{2/2}$) the number of 1-links is constant $\mu_1$ (resp. $\mu_2$) for any $p$, since the formula on $Y$-literals is the same and 1-link $\{v^{1/2},v^z_{i,1}\}$ (resp. $\{v^{2/2},v^z_{i,1}\}$) compensates for $\{v^x_{i,0},v^z_{i,0}\}$.

\begin{figure}[t]
\begin{tikzpicture}[square8/.style={regular polygon,regular polygon sides=8}]
\node[square8,draw=black!25!red,fill=white!80!red, inner sep=0.5em](K) at (3.5,-1.0) {$K_{16}$};
\node[circle,draw=black,fill=white!80!green, inner sep=0.0em](v1) at (1.5,-1.0) {$v^{\frac12}$};
\node[circle,draw=black,fill=white!80!blue, inner sep=0em](v2) at (5.5,-1.0) {$v^{\frac22}$};
\node[] at (-0.7, 0.0) {$N_{i,0}$:};
\node[] at (-0.7, -2.25) {$N_{i,1}$:};
\node[] at (-0.6, 0.65) {$X(C_{i})$:};
\node[] at (0, 0.65) {$\emptyset$};
\node[] at (1, 0.65) {$\emptyset$};
\node[] at (2, 0.65) {$x_1$};
\node[] at (3, 0.65) {$\neg x_1$};
\node[] at (4, 0.65) {$x_1$};
\node[] at (5, 0.65) {$\neg x_2$};
\node[] at (6, 0.65) {$x_3$};
\node[] at (7, 0.65) {$\neg x_3$};
\foreach \x in {0,...,7}
\ifthenelse{\x > 1.5}{
	\pgfmathtruncatemacro{\xx}{\x+1};
	\node[circle, thin, draw=black, inner sep=0mm] (vzi0) at (\x + 0,0) {\hspace*{1mm}};
	\node[circle, thin, draw=black, inner sep=0mm] (vxi0) at (\x + 0,0.3) {\hspace*{1mm}};
	\node[circle, thin, draw=black, inner sep=0mm] (vy1i0) at (\x  -0.3,-0.15) {\hspace*{1mm}};
	\node[circle, thin, draw=black, inner sep=0mm] (vy2i0) at (\x + 0.3,-0.15) {\hspace*{1mm}};
	\node[circle, thin, draw=black, inner sep=0mm] (vzi1) at (\x + 0,-2) {\hspace*{1mm}};
	\node[circle, thin, draw=black, inner sep=0mm] (vxi1) at (\x + 0,-2.3) {\hspace*{1mm}};
	\node[circle, thin, draw=black, inner sep=0mm] (vy1i1) at (\x -0.3,-1.85) {\hspace*{1mm}};
	\node[circle, thin, draw=black, inner sep=-0.2mm] (vy2i1) at (\x + 0.3,-1.85) {\hspace*{1mm}};
	\draw[] (vxi0)--(vy1i0)--(vy2i0)--(vxi0);
	\draw[thick] (vzi0)--(vxi0);
	\draw[] (vxi1)--(vy1i1)--(vy2i1)--(vxi1);
	\ifthenelse{\x < 4.5}{\draw[] (vzi1)--(v1);}{\draw[] (vzi1)--(v2);}
}{
	\node[circle, thin, draw=black, inner sep=0mm] (vz) at (\x + 0, 0) {\hspace*{1mm}};
	\node[circle, thin, draw=black, inner sep=0mm] (vy1) at (\x + -0.15, 0.25) {\hspace*{1mm}};
	\node[circle, thin, draw=black, inner sep=0mm] (vy2) at (\x  -0.15, -0.25) {\hspace*{1mm}};
	\node[circle, thin, draw=black, inner sep=0mm] (vy3) at (\x + 0.3, 0) {\hspace*{1mm}};
	\draw[] (vy1)--(vy2)--(vy3)--(vy1);
	\node[circle, thin, draw=black, inner sep=0mm] (vz) at (\x + 0, -2) {\hspace*{1mm}};
	\node[circle, thin, draw=black, inner sep=0mm] (vy1) at (\x + -0.15, -1.75) {\hspace*{1mm}};
	\node[circle, thin, draw=black, inner sep=0mm] (vy2) at (\x  -0.15, -2.25) {\hspace*{1mm}};
	\node[circle, thin, draw=black, inner sep=0mm] (vy3) at (\x + 0.3, -2) {\hspace*{1mm}};
	\draw[] (vy1)--(vy2)--(vy3)--(vy1);
	\draw[] (vz)--(v1);
}
\draw [dashed, black!25!green] plot [smooth, tension=0.5] coordinates 
	{(-0.4,-1.8) (0,-2.4) (2,-2.4) (2.4,-1.8) (2.0,-0.9) (3.4,-0.2) (3.0,0.5) (2.5,-0.3) (1.5,-0.6) (-0.4,-1.8)};
\draw [dashed, black!25!blue] plot [smooth, tension=0.5] coordinates 
	{(5,0.4) (6,0.4) (6.4,-0.2) (6.1,-0.9) (7.4,-1.8) (7,-2.4) (6.6,-1.8) (5.5,-1.3) (4.4,-1.8) (4,-2.4) (3.6,-1.8) (4.9,-0.9) (4.6,-0.2) (5,0.4) };
\draw [thick, black!25!red] plot [smooth, tension=0.5] coordinates 
	{ (-0.4,-0.2) (0,-0.4) (2.0,-0.5)  (2.7,-2.0) (3.0,-2.5) (3.5,-1.8) (4.0,-1.6) (4.5,-1.8) (5.0,-2.4) (6.1,-2.4) (7.4,-0.2) (7.0,0.4) (6.5,-0.2) (6.0,-1.5)(5.0,-1.5)  (4.05,0.4)  (3.5,-0.3) (2.7,-0.3) (2.1,0.4) (0,0.4) (-0.4,-0.2)};
\end{tikzpicture}
\vspace*{-5mm}
\caption{From \textsc{$\forall\exists$3SAT} to \textsc{Max-Min-3-Partition}: In this proper-3-partition, the attack needs to be a 1-link vertex-cover (giving an independent-set) of node-subset $S^{(p)}$ (red), where $p(x_1)=p(x_2)=0$ and $p(x_3)=1$.}\label{fig:many-tetrads}
\trimafterfigure
\end{figure}
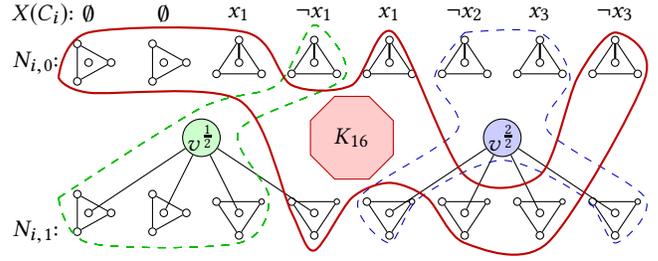

We show that in our construct, any 3-partition which is not a \emph{proper}-3-partition
does trivially satisfy the complement question above. 
First, let us reason as if all three node-subsets were cliques of $\Lambda$-links.
Crucially, in a node-subset of size $\nu$, the number of links ${\nu\choose 2}$ is quadratic. 
Therefore, the largest node-subsets will be the first attacked, and the only way $\pi_{-M}$ contains as many as ${2m\choose 2}+2{m+1\choose 2}$ $\Lambda$-links is if the node-subsets of $\pi$ had sizes $3m$, $m+1$ and $m+1$.
Second, assume $\Lambda$-links are missing in some node-subsets. 
Then, an attack would focus on more connected subsets and $\pi_{-M}$ cannot contain as many as ${2m\choose 2}+2{m+1\choose 2}$ $\Lambda$-links. 
Therefore, 3-partition $\pi$ must consist in $\Lambda$-link cliques of size $3m$, $m+1$ and $m+1$.
If the largest did not follow consistently some instantiation $p:X\rightarrow\{0,1\}$, 
then some $\Lambda$-links would be missing (see ($\star$)). Also, the only way to obtain two $\Lambda$-linked cliques of size $m+1$ on $N\setminus S^{(p)}$ is by $S^{1/2}$ and $S^{2/2}$.
We also know that $S^{1/2}$ and $S^{2/2}$ contain $\mu_1+\mu_2$ 1-links.

Crucially, attack $M$ always occurs where it does the largest damage w.r.t. $\Lambda$-links: on node-subset $S^{(p)}$, and the number of remaining $\Lambda$-links is ${2m\choose 2} + 2{m+1\choose 2}$. 
Given a proper-3-partition, what could make the inequality false would be a surviving 1-link in $S^{(p)}\setminus M$. Consequently, condition $\exists M,W(\pi_{-M})\leq\theta-1$ amounts to a $2\alpha$ node attack $M$  that covers every 1-link in $\bigcup_{i=1}^{i=\alpha} N_{i,p(X(C_i))}$. 
A crucial observation is that we necessarily attack/cover exactly two nodes per tetrad $N_{i,j}$, since each tetrad contains a triangle. 
In negative tetrads $N_{i,0}$, because of 1-link $\{v^{x}_{i,0},v^{z}_{i,0}\}$, one of these nodes has to be $v^{x}_{i,0}\in M$.  
In positive tetrads $N_{i,1}$, since node $v^{x}_{i,1}$ is not involved in other 1-links than the triangle, 
choosing both $v^{y}_{i,1}$ and $v^{y'}_{i,1}$ in 1-link cover $M$ is the best choice.
As in $\textsc{3SAT}\leq\textsc{IndependentSet}$, this amounts to a 1-link-independent-set $\overline{M}= S^{(p)}\setminus (K\cup M)$ with size $2\alpha$ and two nodes per tetrad $N_{i,j}$: first, node $v^{z}_{i,j}$, second if $j=0$ then $v^{y}_{i,0}$ xor $v^{y'}_{i,0}$, otherwise if $j=1$ then $v^{x}_{i,1}$.

(yes$\Rightarrow$yes) Assume that for every $\tau_x:X\rightarrow\{0,1\}$, there exists $\tau_y:Y\rightarrow\{0,1\}$ such that in every clause $C_i$ with $\tau_x(X(C_i))=0$, a $Y$-literal is made true by instantiation $\tau_y$. We show that given any proper-3-partition $\{S^{(p)},S^{1/2},S^{2/2}\}$, in  $S^{(p)}\setminus K=\bigcup_{i=1}^{i=\alpha} N_{i,p(X(C_i))}$, there exists a 1-link-independent-set $\overline{M}$ of size $2\alpha$, as below. 
Taking $\tau_x\equiv p$, let $\tau_y:Y\rightarrow\{0,1\}$ be as above mentioned.
Then, 
\begin{eqnarray*}
\overline{M} &=& \bigcup_{i\in[\alpha]}\left\{
\begin{array}{ll}
\text{if } p(X(C_i))\!=\!0 \text{:} & \{v^z_{i,0}, \text{one } v^y_{i,0}\mid\tau_y(\ell_i^y)=1\}\\
\text{if } p(X(C_i))\!=\!1 \text{:} & \{v^z_{i,1}, v^x_{i,1}\}
\end{array}\right.
\end{eqnarray*}
is a 1-link-independent-set of size $2\alpha$: node $v^y_{i,0}$ exists since instantiation $\tau_y$ gives at least one true literal per clause where $\tau_x(X(C_i))=0$, and nodes are not 1-linked (no literal contradiction).

(yes$\Leftarrow$yes) Assume that for any $\tau_x\equiv p:X\rightarrow\{0,1\}$, a 1-link-independent-set $\overline{M}$ with size $2\alpha$ exists in node-subset $S^{(p)}\setminus K=\bigcup_{i=1}^{i=\alpha} N_{i,p(X(C_i))}$. 
Then, nodes $v^y_{i,0}\in\overline{M}$ consistently define $\tau_y:Y\rightarrow\{0,1\}$ that makes any clause $C_i$ true whenever $\tau_x(X(C_i))=0$.

Crucially, we also include clauses without any $X$-literal in the same construct. 
Assume w.l.o.g. that there are less than $\alpha/2$ such $Y$-clauses, within the first indexes in $[\alpha]$. To any $Y$-clause $C=\ell^{y}_{i}\vee\ell^{y'}_{i}\vee\ell^{y''}_{i}$, one associates 
two tetrads $N_{i,j}=\{v^{y}_{i,j}, v^{y'}_{i,j}, v^{y''}_{i,j}, v^{z}_{i,j}\}$, $j\in\{0,1\}$. 
For $C_{i},C_{i'}$ $Y$-clauses, between $N_{i,0}$ and $N_{i',1}$ weights are zero. 
Negative tetrads $N_{i,0}$ are fully $\Lambda$-linked inside, between themselves, with previous tetrads of  one $X$-variable and set $K$.
Positive tetrads $N_{i,1}$ are fully $\Lambda$-linked inside, between themselves and with $v^{1/2}$.
Given a $Y$-clause $C$, we define $X(C)=\emptyset$.
For proper-3-partitions, we extend $p(\emptyset)=0$; 
hence in $\{S^{(p)},S^{1/2},S^{2/2}\}$, for $C_i$ a $Y$-clause, one has $N_{i,0}\subseteq S^{(p)}$ and $N_{i,1}\subseteq S^{1/2}$.
Similarly, in any $Y$-clause tetrad $N_{i,j}$, there are 1-links
$\{\{v^{y}_{i,j}, v^{y'}_{i,j}\},
\{v^{y'}_{i,j}, v^{y''}_{i,j}\},
\{v^{y''}_{i,j}, v^{y}_{i,j}\}\}$, 
(optional 1-links $\{v^{z}_{i,1},v^{1/2}\}$), 
and whenever two $Y$-literals are complementary.
It follows that the same proof holds.
\end{proof}


\section{Related Work}

Partitioning of a set into (non-empty) subsets may also be referred as coalition structure formation of a set of agents into coalitions.
When a number of coalitions $k$ is required and there are synergies between vertices/agents, this problem is referred as $k$-cut, or $k$-way partition, where one minimizes the weight of edges/synergies between the coalitions, or maximizes it inside the coalitions.
For positive weights and $k\geq 3$, this problem is NP-complete \cite{Dahlhaus:1992:CMC:129712.129736}, when one vertex is fixed in each coalition.
For positive weights and fixed $k$, a polynomial-time $O(n^{k^2}T(n,m))$ algorithm exists \cite{10.2307/3690374}, when no vertex is fixed in coalitions, and where $T(n,m)$ is the time to find a minimum $(s,t)$ cut on a graph with $n$ vertices and $m$ edges.
When not too many negative synergies exist (that is, negative edges can be covered by $O(\log(n))$ vertices), an optimal $k$-partition can be computed in polynomial-time \cite{SLESS2018217}.

 \bibliographystyle{ACM-Reference-Format}  
 \bibliography{ourbiblio}  

\end{document}